\documentclass[a4paper,USenglish]{lipics-v2021}
\nolinenumbers	
\usepackage{graphicx}
\usepackage{subcaption}
\usepackage{balance}
\usepackage[linesnumbered,ruled,vlined]{algorithm2e}
\SetKwInOut{Parameter}{Parameters}
\SetKwInOut{Output}{Output}
\usepackage{multicol}  
\usepackage{amsmath}
\usepackage{amsthm}

\usepackage{tikz}
\usetikzlibrary{positioning, arrows.meta, shapes.misc}
\usetikzlibrary{patterns}

\usepackage{mathtools}
\DeclarePairedDelimiter\floor{\lfloor}{\rfloor}
\DeclarePairedDelimiter{\ceil}{\lceil}{\rceil}

\usepackage{hyperref}
\usepackage{xcolor}
\hypersetup{
    colorlinks,
    linkcolor={red!50!black},
    citecolor={blue!50!black},
    urlcolor={blue!80!black}
}

\usepackage{soul}

\def\p{\hat{p}}
\def\JoinCom{\mbox{\sf JoinComm}}
\def\Constr{\mbox{\sf Constr}}
\def\lt{\left}
\def\rt{\right}
\def\inline#1:{\par\vskip 7pt\noindent{\bf #1:}\hskip 10pt}

\title{Time-Optimal and Energy-Efficient Deterministic Consensus}

\author{Shachar Meir}
{Weizmann Institute of Science, Rehovot, Israel \and \url{https://shacharmeir007.github.io}}
{shachar.meir@weizmann.ac.il}
{0009-0003-5007-047X}
{}
\author{Hugo Mirault}
{Augusta University, Augusta, Georgia, USA}
{hmirault@augusta.edu}
{0009-0008-5885-8372}
{}
\author{David Peleg}
{Weizmann Institute of Science, Rehovot, Israel \and \url{https://www.weizmann.ac.il/math/peleg/}}
{david.peleg@weizmann.ac.il}
{0000-0003-1590-0506}
{}
\author{Peter Robinson}
{Augusta University, Augusta, Georgia, USA}
{perobinson@augusta.edu}
{0000-0002-7442-7002}
{}
\authorrunning{S. Meir \and H. Mirault \and D. Peleg \and P. Robinson}

\Copyright{Shachar Meir, Hugo Mirault, David Peleg, Peter Robinson} 

\funding{Hugo Mirault and Peter Robinson were supported in part by National Science Foundation (NSF) Grant No. CCF-2402836.}

\relatedversiondetails[cite=10.1145/3732772.3733554]{Previous Brief Annoucement}{https://dl.acm.org/doi/abs/10.1145/3732772.3733554}

\ccsdesc[500]{Theory of computation~Distributed algorithms}

\keywords{Distributed computing, Crash faults, Consensus, Energy complexity, Sleeping
model}

\EventEditors{Andrei Arusoaie, Emanuel Onica, Michael Spear, and Sara Tucci-Piergiovanni}
\EventNoEds{4}
\EventLongTitle{29th International Conference on Principles of Distributed Systems (OPODIS 2025)}
\EventShortTitle{OPODIS 2025}
\EventAcronym{OPODIS}
\EventYear{2025}
\EventDate{December 3--5, 2025}
\EventLocation{Ia\c{s}i, Romania}
\EventLogo{}
\SeriesVolume{361}
\ArticleNo{15}

\begin{document}

\maketitle
\begin{abstract}
We study fault-tolerant consensus in a variant of the synchronous message passing model, where, in each round, every node can choose to be awake or asleep.
This is known as the sleeping model (Chatterjee, Gmyr, Pandurangan PODC 2020) and defines the awake complexity (also called \emph{energy complexity}), which measures the maximum number of rounds that any node is awake throughout the execution.
Only awake nodes can send and receive messages in a given round and all messages sent to sleeping nodes are lost.
We present new deterministic consensus algorithms that tolerate up to $f<n$ crash failures, where $n$ is the number of nodes.
Our algorithms match the optimal time complexity lower bound of $f+1$ rounds.
For multi-value consensus, where the input values are chosen from some possibly large set, we achieve an energy complexity of $\mathcal{O}(\lceil f^2 / n \rceil)$ rounds, whereas for binary consensus, we show an algorithm to achieve $\mathcal{O}(\lceil f / \sqrt{n} \rceil)$ energy complexity.
\end{abstract}

\bigskip
\noindent
{\bf Eligibility:}
The paper is eligible for the best student paper award.

\bigskip

\section{Introduction} \label{sec:intro}
Reaching agreement in a network of distributed processes in the presence of faults is a fundamental problem in distributed computing with numerous applications, including blockchain networks and synchronization problems; e.g., see \cite{DBLP:journals/csur/00310J23,attiyawelch,Lyn96}.
In the consensus problem~\cite{pease1980reaching,lamport1982byzantine}, there is a set of players (or nodes) that communicate with each other by sending point-to-point messages. Each player holds an initial value, and the goal is that all of them decide on the same value eventually.
 
While the problem can be solved in just a single round of communication assuming that all players operate correctly, achieving consensus becomes significantly more challenging when taking into account that some of the players may unexpectedly fail throughout the execution.
The seminal result of \cite{FLP85} showed that consensus is impossible with a deterministic algorithm in asynchronous networks, when even just a single player may crash.
Moreover, \cite{dolevstrong} proved that $f+1$ rounds are indeed necessary in the synchronous setting, when considering up to $f$ crash failures, and a simple algorithm matches this lower bound on the time complexity. 

Even though the time complexity of deterministic consensus has been resolved, the actual number of rounds that a player needs to actively participate in the protocol is less clear. 
This motivates studying consensus in the sleeping model, where, in every round, each player chooses to either be asleep or awake, and the \emph{energy complexity} of the algorithm is defined as the maximum number of awake rounds over all players in the worst case.
The sleeping model was first introduced in \cite{chatterjee2020sleeping} for studying the (average) energy complexity\footnote{The terms ``energy complexity'' and ``awake complexity'' are used interchangeably in the literature. In this work, we use energy complexity.} of computing a maximal independent set and has since received significant attention recently; e.g., see \cite{DBLP:conf/podc/DufoulonMP23,DBLP:conf/podc/0001T24,DBLP:conf/podc/0001P23,DBLP:conf/icdcs/HouraniPR22,DBLP:conf/wdag/BarenboimM21}. 
The sleeping model can be interpreted as a simplified variant of the class of models studied in \cite{DBLP:conf/podc/ChangDHHLP18,DBLP:conf/podc/ChangDHP20}, where the main focus is on wireless networks.

In \cite{pass2017sleepy}, Pass and Shi introduced the \emph{sleepy model of consensus} in the context of Blockchain networks with Byzantine players, where the players exhibit sporadic participation: A player can be either online (alert) or offline (asleep),
and it is possible that only a fraction of the players participate in the computation at any given point in time. 
A player's current participation status may change as determined by the adversary.
Moreover, once an asleep player becomes online, it is guaranteed to receive all pending messages that were addressed to it while sleeping.
The sleepy model has become a standard model for studying blockchain protocols under dynamic participation and has been studied extensively since its inception; e.g., see \cite{efron2025fully,malkhi2023towards,momose2022constant,neu2025limits,d2025tob,d2024asynchrony,efron2025optimal}. For instance, a constant-latency protocol in this setting is given in \cite{momose2022constant}, and the constant factor was further reduced in \cite{malkhi2023towards}. A variant of the sleepy model where the set of players is infinite is considered in \cite{losa2023consensus}. 
A generalized version of the sleepy model was studied in \cite{d2024recent}.
Note that the assumptions of the sleepy model are different from the sleeping model of~\cite{chatterjee2020sleeping} considered here, which assumes that the sleeping schedule of the players is entirely under the control of the algorithm and all messages sent to a currently sleeping player are lost.

\subsection{Our Contributions}
We establish the first energy efficient algorithms for consensus under crash faults in the sleeping model, while still obtaining optimal time complexity.
In more detail, we present the following new deterministic algorithms:
\begin{itemize}
    \item Multi-value consensus can be solved with an energy complexity of $\mathcal{O}(\lceil\frac{f^2}{n}\rceil)$ rounds, an optimal time complexity of $f+1$, while sending at most $\mathcal{O}(f^3+nf)$ messages.
    \item Binary consensus can be solved with an energy complexity of $\mathcal{O}(\lceil \frac{f}{\sqrt{n}}\rceil)$ rounds, an optimal time complexity of $f+1$, and with a message complexity of $\mathcal{O}(nf)$.
\end{itemize}
We point out that none of the prior works in the sleeping model considers faulty nodes, and thus our results also shed some light on the impact of faults on the energy complexity.

\section{Model and Preliminaries}
We consider the standard synchronous message-passing
model of distributed computing, where $n$ players $p_0,\ldots,p_{n-1}$ are connected in a clique network.
Each player can choose to either \emph{sleep} or be \emph{awake} in each round, which is known as the \emph{sleeping model}.
If a player $p$ is asleep in some round, then any message sent to $p$ in this round is lost; in particular, $p$ does \emph{not} receive these messages upon its next wake-up.
In this work, we extend the above model with \emph{crash faults}, which means that an adversary can crash up to $f$ players throughout the execution, for some parameter $f<n$.
As we only consider deterministic algorithms, we assume that the adversary is omniscient. 
If a player $p$ \emph{crashes} in round $r$, then only an arbitrary subset of messages sent by $p$ are guaranteed to be delivered.
Moreover, $p$ does not perform any computation in any future round after crashing.
Any player who does not crash in rounds $1,\dots,r-1$ is said to be \emph{alive in round $r$}
If a player does not crash until it has decided, we say that it is \emph{non-faulty}.

It is possible in the sleeping model that there are certain rounds in which no players are awake, e.g., because the adversary crashed the few player that tried to wake up. 
Consequently, any algorithm that solves consensus in optimal time needs to implement an awake schedule that
ensures sufficient progress towards a common decision in each round.

\inline Complexity measures:
We quantify the performance of our algorithms in terms of the \emph{message complexity}, i.e., the total number of messages sent in the worst case, and the \emph{time complexity}, which captures the worst-case number of rounds until all players have terminated.
In the sleeping model, we consider an additional measure:
\begin{itemize} 
\item the \emph{(worst case) energy complexity} is the maximum number of rounds (taken over all players) that a player is awake in the worst case execution;
\end{itemize}

\inline Consensus:
Every player starts with an input value that is assigned from a set $D$ of integers. In Section~\ref{sec:multi}, we assume that $D$ is of size at most polynomial in $n$, whereas in Section~\ref{sec:binary}, we restrict $D$ to be $\{0,1\}$.
An algorithm solves consensus if, in any execution and for any input assignment, the following properties hold:
\begin{itemize}
\item \textit{Termination}: Every non-faulty player eventually decides on some value. 
\item \textit{Agreement}: If some player decides on $v$ in an execution $\alpha$, then every other player who decides must also decide on $v$.
\item \textit{Validity}: Any decision value must be the input value of some (possibly faulty) player. 
\end{itemize}

\section{Multi-value Consensus} \label{sec:multi}
In this section, we give an algorithm for the setting, where the input values of the players are chosen from an arbitrary subset of the integers.
We present the detailed pseudo-code in Algorithm~\ref{alg:multi}.

\subsection{Description of the Algorithm} \label{sec:desc}
At the beginning of round $1$, we fix $f$ static committees $C_1,\dots,C_f$, each containing $f+1$ members. The players use their IDs to determine which committee to join.
That is, each player locally executes a function
\JoinCom($x,y,n$), which 
constructs $x$ committees of size $y$, using all $n$ players, such that each player's ID is in at most $\lceil\frac{f(f+1)}{n}\rceil$ committees. Note that the procedure's third parameter, i.e, the number of players used to construct the committees, will be utilized carefully in Section \ref{sec:binary}. For now, the reader should focus on the choice of the first two parameters.

\begin{algorithm}
\small
\Parameter{$a:$ number of committees\\$b:$ size of committees\\ $c:$ number of players to use (starting from the first one) }
$C_1,\cdots,C_a \gets \emptyset$\\
\ForEach{$1\leq i\leq a \cdot b$}
{$j \gets i \mod c$\\
$k \gets \ceil{i/b}$\\
Add $p_j$ in set $C_k$}
Return  $C_1,\cdots,C_a$
\caption{Procedure \JoinCom$(a,b,c)$ assigns the first $c$ players to $a$ committees, each of size $b$.}
\label{alg:joincommittees}
\end{algorithm}

In Phase~1, which consists of round $1$, all players awake simultaneously and send their initial values to committee $C_1$. Each player $p \in C_1$ keeps track of the maximum received value in a variable $Y$ that is initialized to $p$'s input.
Phase~2 consists of $f-1$ rounds, where all members of committee $C_r$ are awake in rounds $r$ and $r+1$. In round $r$, they receive the values sent by committee $C_{r-1}$ and each member computes its new maximum value locally. In round $r+1$, each player in $C_r$ forwards this value to the next committee $C_{r+1}$.
Finally, Phase 3 consists of round $f+1$, where everyone wakes up, and the committee members of $C_f$ broadcast their values, and everyone decides on the maximum value that they have seen. See Figure~\ref{fig:multi} for a high-level overview of this process.

\begin{algorithm}[t] \small
\begin{multicols}{2}
\tcp{Phase 1 (round 1)}
\texttt{Player $p$ wakes up in round 1:}\\
$C_1,\cdots,C_{f+1} \gets$ \JoinCom($f,f+1$)\\
$Y \gets X$\\
Send $Y$ to every player in $C_1$\\
\textbf{upon} receiving values $Y'$ \textbf{do}:\\
\phantom{-------}$Y \gets \max(Y \cup (\bigcup Y'))$\\
\medskip
\tcp{Phase 2 (rounds $2,\dots,f$)}
\textbf{For each} round $r \in [2,f]$ \textbf{do}: \\
\If{$p \in C_{r-1}$} 
{\texttt{Wake up in round $r$}\\
Send $Y$ to every player in $C_r$}
\If{$p \in C_{r}$}
{\texttt{Wake up in round $r$}\\
\textbf{upon} receiving values $Y'$ \textbf{do}:\\
\phantom{-------}$Y \gets \max(Y \cup (\bigcup Y'))$}
\medskip
\tcp{Phase 3 (round $f+1$)}
\texttt{Player $p$ wakes up in round $f+1$}\\
\If{$p \in C_{f}$}
{Broadcast $Y$ to all\\}
\textbf{upon} receiving values $Y'$ \textbf{do}:\\
\phantom{-------}$Y \gets \max(Y \cup (\bigcup Y'))$\\
Decide on $Y$ at the end of round $f+1$\\
\end{multicols}
\caption{Multi Value Consensus. Code for player $p$.}
    \label{alg:multi}
\end{algorithm}

\begin{figure}[t]
    \centering 
    \resizebox{0.75\linewidth}{!}{\begin{tikzpicture}
\tikzset{
player/.style={
    draw=gray,
    thin,
    dashed,
    opacity=0.6
},
	message/.style={
    ->,
    ultra thick,
    >=Latex,
    shorten >=1pt, 
    shorten <=1pt, 
    rounded corners,
    color=cyan!70
},
    time/.style={thick, dashed}
}

\def\t{15}

\draw[thick, rounded corners] (0.5,-8.5) rectangle (2,2.5);
\node at (1.5, 3.3) {Phase 1};
\draw[thick, rounded corners] (2,-8.5) rectangle (13,2.5);
\node at (7.5, 3.3) {Phase 2};
\draw[thick, rounded corners] (13,-8.5) rectangle (14.5,2.5);
\node at (13.5, 3.3) {Phase 3};

\node[anchor=east] at (0,2) {$P_1$};
\draw[player] (0,2) -- (\t,2);
\node[anchor=east] at (0,1) {$\cdots$}; 
\draw[player] (0,1) -- (\t,1);
\node[anchor=east] at (0,0) {$P_{f+1}$};
\draw[player] (0,0) -- (\t,0);
\node[anchor=east] at (0,-1) {$P_{f+2}$};
\draw[player] (0,-1) -- (\t,-1);
\node[anchor=east] at (0,-2) {$\cdots$};
\draw[player] (0,-2) -- (\t,-2);
\node[anchor=east] at (0,-3) {$P_{2f+2}$};
\draw[player] (0,-3) -- (\t,-3);
\node[anchor=east] at (0,-4) {$P_{2f+3}$};
\draw[player] (0,-4) -- (\t,-4);
\node[anchor=east] at (0,-5) {$\cdots$};
\draw[player] (0,-5) -- (\t,-5);
\node[anchor=east] at (0,-6) {$P_{3f+3}$};
\draw[player] (0,-6) -- (\t,-6);
\node[anchor=east] at (0,-7) {$\cdots$};
\draw[player] (0,-7) -- (\t,-7);
\node[anchor=east] at (0,-8) {$P_{n}$};
\draw[player] (0,-8) -- (\t,-8);

\draw[thick, dashed, rounded corners, fill=blue!20, fill opacity=0.3] (1.25, -0.25) rectangle (3.25, 2.25);
\node[text=blue] at (3.5, 1.2) {$C_1$};
\draw[thick, dashed, rounded corners, fill=blue!20, fill opacity=0.3] (2.75, -3.25) rectangle (4.25, -0.75);
\node[text=blue] at (2.4, -1.8) {$C_2$};

\draw[thick, dashed, rounded corners, fill=blue!20, fill opacity=0.3] (3.75, -6.25) rectangle (5.25, -3.75);
\node[text=blue] at (3.4, -4.8) {$C_3$};

\draw[thick, dashed, rounded corners, fill=blue!20, fill opacity=0.3] (4.75, -7.25) rectangle (6.25, -6.75);
\node[text=blue] at (4.4, -6.8) {$C_4$};

\draw[thick, dashed, rounded corners, fill=blue!20, fill opacity=0.3] (5.75, -0.25) rectangle (7.25, 2.25);
\node[text=blue] at (5.4, 1.2) {$C_{h}$};

\draw[thick, dashed, rounded corners, fill=blue!20, fill opacity=0.3] (6.75, -3.25) rectangle (8.25, -0.75);
\node[text=blue] at (6.4, -1.8) {$C_{h+1}$};

\draw[thick, dashed, rounded corners, fill=blue!20, fill opacity=0.3] (7.75, -6.25) rectangle (9.25, -3.75);
\node[text=blue] at (7.4, -4.8) {$C_{h+2}$};

\draw[thick, dashed, rounded corners, fill=blue!20, fill opacity=0.3] (8.75, -7.25) rectangle (10.25, -6.75);
\node[text=blue] at (8.4, -6.8) {$C_{h+3}$};

\draw[thick, dashed, rounded corners, fill=blue!20, fill opacity=0.3] (9.75, -0.25) rectangle (11.25, 2.25);
\node[text=blue] at (9.4, 1.2) {$C_{f-2}$};

\draw[thick, dashed, rounded corners, fill=blue!20, fill opacity=0.3] (10.75, -3.25) rectangle (12.25, -0.75);
\node[text=blue] at (10.4, -1.8) {$C_{f-1}$};

\draw[thick, dashed, rounded corners, fill=blue!20, fill opacity=0.3] (11.75, -6.25) rectangle (13.25, -3.75);
\node[text=blue] at (11.4, -4.8) {$C_{f}$};

\foreach \step in {2,0,-1,-2,-3,-4,-5,-6,-7,-8}{\draw[message,shorten >=0px,-{}] (0.5, \step) -- (1.75, \step);}
\draw[message,shorten >=0px] (0.5, 1) -- (1.75, 1);
\foreach \step in {2,-8}{\draw[message,shorten >=10px] (1.75, \step) -- (1.75, 1);}
\draw[message] (3, 1) -- (3, -2);
\draw[message] (4, -2) -- (4, -5);
\draw[message] (5, -5) -- (5, -7);
\draw[message,dashed] (6,-7) -- (6,1);
\draw[message] (7, 1) -- (7, -2);
\draw[message] (8, -2) -- (8, -5);
\draw[message] (9, -5) -- (9, -7);
\draw[message,dashed] (10,-7) -- (10,1);
\draw[message] (11, 1) -- (11, -2);
\draw[message] (12, -2) -- (12, -5);
\draw[message,shorten >=0px,-{}] (12.5, -5) -- (13.5, -5);
\draw[message,shorten >=0px,-{}] (13.5, 2) -- (13.5, -8);
\foreach \step in {2,1,0,-1,-2,-3,-4,-5,-6,-7,-8}{\draw[message] (13.5, \step) -- (15, \step);}

\end{tikzpicture}}
    \caption{High-level Overview of Algorithm~\ref{alg:multi}}
    \label{fig:multi}
\end{figure}
\begin{theorem} \label{thm:multi}
Suppose that each player starts with an input value from a set of integers $D$.
Algorithm~\ref{alg:multi} solves consensus in $f+1$ rounds, sends $\mathcal{O}(f^3+nf)$ messages of size $\mathcal{O}(\log |D|)$ bits, and achieves an energy complexity of $\mathcal{O}(\lceil{f^2}/{n}\rceil)$ rounds, if at most $f$ players crash during the execution.
\end{theorem}

\subsection{Proof of Theorem~\ref{thm:multi}}
We first show that the algorithm correctly solves consensus.
Every correct player decides on its variable $Y$ after $f+1$ rounds, which ensures the claimed bound on the time complexity.
By the code of the algorithm, $Y$ is only updated to the input values of other players, and thus we obtain the following.
\begin{lemma} \label{lem:validity_term}
Algorithm~\ref{alg:multi} satisfies validity and terminates in $f+1$ rounds.
\end{lemma}
To show that the algorithm satisfies agreement, we need to prove several technical lemmas. 
Consider the set $D$ of values decided by the players after round $f+1$ in some given execution, and let $\delta = \max(D)$. 
We say that \emph{a player $p$ has value $v$} if $p$'s variable $Y=v$ at that point.
\begin{lemma} \label{lem:nodec}
Suppose that there exists some round $r \in [1,f+1]$ with the property that there is an awake player $p$, who is alive in round $r$ and sends a message with value $v_p > \delta$. Then no player decides on $\delta$. 
\end{lemma}
\begin{proof} 
Consider any round $r$ that satisfies the premise of the claim.
It follows that $p$ successfully sends its message to all destinations.
In the case that $r=f+1$, this means that $p \in C_f$, and every player receives $p$'s value in round $f+1$, which is strictly greater than $\delta$. 
Hence, none of the players will decide on $\delta$.
On the other hand, if $r\le f$, then every alive member $p' \in C_{r}$ receives the message from $p$ and updates its variable $Y$ such that $Y>\delta$. Since not all members of $C_{r}$ can crash, it follows that there is some alive player in $C_{r}$, who is awake and alive in round $r+1$ and who sends a value strictly greater than $\delta$. 
Extending this argument inductively for all rounds $r,\dots,f+1$, completes the proof of the lemma.
\end{proof}
\begin{lemma} \label{lem:max}
Consider any round $r \in [1,f]$ and let $m \ge \delta$ be the maximum value held by any player in $C_{r}$ at the start of round $r+1$. 
Then, every alive player has a value of at most $m$ at the beginning of round $r+1$.
\end{lemma}
\begin{proof}
The statement trivially holds for any alive player in $C_r$.
Now, consider some alive player $p \notin C_{r}$ and assume, towards a contradiction, that $p$'s value is $v>m$ when it starts round $r+1$.
First, consider the case where $v$ was $p$'s input value. 
This means that $p$ sent $v$ to every player in $C_1$ initially.
Since $p$ did not crash, every $p' \in C_1$ must have set $Y$ to some value $v_{p'} \ge v$ at the end or round $1$ and subsequently sends $v_{p'}$ to everyone in $C_2$ in round $2$.
Recall that every committee has at least one member who does not crash.
Thus, it follows by induction that every player $q \in C_{r'}$ ($r' \ge 2$) will receive a message with a value no smaller than $v$ in round $r'$ from some alive player in $C_{r'-1}$ and update its own value accordingly. 
Consequently, $q$ has set $Y$ to a value strictly greater than $m$ at the start of round $r'+1$.
In particular, this statement holds for $r'=r$, which contradicts the assumption that $m$ was the maximum value among players in $C_r$.
Now suppose that $v$ was not $p$'s input value, which means that $p$ must have been a member of some committee $C_{r'}$ and adopted $v$ as its value in some round $r'\le r$.
Since $p$ did not crash in round $r'+1$ and recalling that $v > \delta$, applying Lemma~\ref{lem:nodec} yields a contradiction to the assumption that some player decided on $\delta$. 
\end{proof}
\begin{lemma} \label{lem:agreement}
Algorithm~\ref{alg:multi} satisfies agreement.
\end{lemma}
\begin{proof}
Consider the earliest round $r \in [1,f+1]$ in which no player crashes.
First, assume that $r=1$. 
Then, it follows from Lemma~\ref{lem:nodec} that no player has an input greater than $\delta$. 
Thus, it follows from validity (see Lem.~\ref{lem:validity_term}) that $\delta$ was the maximum input value, and hence every player in $C_1$ adopts $\delta$ in round $1$.
It follows that all alive members of every committee will have value $\delta$, and thus every player decides on $\delta$ eventually.  
On the other hand, if $1 < r < f+1$, then Lemma~\ref{lem:nodec} and the fact that someone decides on $\delta$ tell us that every alive player in $C_{r-1}$ must have sent a value of at most $\delta$ in round $r$.
Moreover, Lemma~\ref{lem:max} ensures that all alive players (whether they are sleeping or awake) have a value of at most $\delta$ at the start of round $r+1$; let $V$ denote this set of values.
If $\max V < \delta$, then Lemma~\ref{lem:max} implies us that every alive player has a value strictly less than $\delta$.
Since values are updated using the maximum rule, this contradicts the assumption that some player decided on $\delta$, and hence it follows that some alive player $p \in C_r$ has the value $\delta$, which is the maximum among all committee members.
And, applying Lemma~\ref{lem:max} guarantees that no greater values are held by players not in $C_r$.
Since $p$ does not crash, a simple inductive argument shows that $\delta$ is the only possible decision value.
\end{proof}
It remains to prove the claimed bound on the energy and message complexity:
The first and the last phases require every player to be awake for a constant number of rounds. 
The second phase requires the players in each committee to be awake only in the two rounds during which their respective committee needs to receive or send messages. Since each player belongs to at most $\lceil \frac{(f+1)\cdot f}{n} \rceil$ committees, the energy complexity is $\mathcal{O}(\lceil{f^2}/{n}\rceil)$.
During Phases 1 and 3, the algorithm sends $\mathcal{O}(f\cdot n)$ messages. 
In each round $r$ of Phase~2, all players in committee $C_{r-1}$ send a message to every player in committee $C_r$, which amounts to $\mathcal{O}(f^2)$ messages for each of the $f-1$ rounds, and thus the total number of messages is $\mathcal{O}(f^3+f\cdot n)$.

\section{Binary Consensus} \label{sec:binary}
In this section, we show how to improve over the energy complexity bounds obtained in Section \ref{sec:multi}, by restricting the domain of the input values to a single bit, while keeping the optimal time complexity of $f+1$ rounds.
Recall that, in Algorithm \ref{alg:multi}, the energy complexity was dominated by the number of committees that each player is a member of, which in turn depends on the size of the individual committees, i.e., $f+1$. 
Here, we explore the natural strategy of reducing the committee size to achieve a lower energy complexity. 
While conceptually simple, having smaller committees entails that we can no longer rely on each committee having at least one alive member, requiring a player to spread its value among multiple committees instead. 
Moreover, it might create an opening for the adversary to strategically crash players, one at each round, such that eventually a round is reached where all the members of a committee (or in general all the players that should receive messages) have already crashed. This allows the adversary to abstain from crashing during that particular round without any player receiving new information, which makes it impossible for the algorithm to terminate after $f+1$ rounds.
We compensate for this lack of information spreading across committees by using players aware of decisive information (here, players holding 1) as propagating units for further committees.
Assume players may choose 1 over 0 when they have the choice (i.e., they know at least 1 player holds each of those values as their initial value). This situation allows a player to know that everyone may decide 1 when he receives the value, and/or he will crash before the end of the procedure.

\subsection{Description of the Algorithm}
We now give a high-level overview of our algorithm; see  Algorithm \ref{alg:binary} for the detailed pseudo code.

\begin{algorithm}[t]
\small
\begin{multicols}{2}
\tcp{Phase 1 (round 1)}
\texttt{player $p$ wakes up in round 1:}\\
$h\gets \min \{f, n'-\sqrt{n'}+1\}$\\

\tcp{$h$ = first round of Phase 3 or 4, \\
depending on the relation between \\
$f$ and $n'$}
Locally compute first batch of \\
committees for Phases 1 and 2:
$C_1,\cdots,C_{h-1}\gets$\JoinCom($h-1, \sqrt{n'},n'$)\\
Locally compute second batch of \\ committees 
: $C_{h},\cdots,C_{f}\gets$\JoinCom($f -h+1, f+1,n$)\\
$Z\gets 0$; $Y \gets 0$; $T \gets 0$\\
\If{$X = 1$}
{
$Y\gets 1$; $T \gets \ceil {\frac{f+1}{\sqrt{n'}}}$ \label{alglin:binary_validity}\\
Send $1$ to every player in  committee $C_1$
}
\If{$p$ received value $1$ and $Y = 0$}
{$Y \gets 1$; $T \gets \lceil \frac{f+1}{\sqrt{n'}} \rceil$}
\smallskip
\tcp{Phase 2 (rounds $2,\ldots,h-1$)}
\textbf{For each} round $r \in [2,h-1]$ \textbf{do}: \\
\If{$T>0$\label{algline:impbinaryLoopCondition}} 
{\texttt{player $p$ wakes up in round $r$}\\
Send $1$ to every player in $C_r$\\
$T \gets T - 1$
} 

\If{$p \in C_{r}$}
{\texttt{player $p$ wakes up in round $r$}\\
\If{$p$ received value $1$ and $Y = 0$}
{$Y \gets 1$; $T \gets \lceil \frac{f+1}{\sqrt{n'}} \rceil$\\}}

\medskip
\tcp{Phase 3 (rounds $h, \ldots, f-1$)}
\textbf{For each} round $r \in [n'-\sqrt{n'} +1,f-1]$ \textbf{do}: \\
\tcp{r = first round of the phase}
\If{$Y=1$ and $r=n'-\sqrt{n'} +1$
}{
\texttt{player $p$ wakes up in round $r$} 
\\
Send 1 to every player in $C_{r}$.
}
\If{$T>0$}{
    \texttt{player $p$ wakes up in round $r$} \\
    Send 1 to every player in $C_r$.\\
    $T\gets T-1$
    }
\If{$p \in C_r$}{
    \texttt{player $p$ wakes up in round $r$} \\
    \If {$p$ received value $1$ \and $Z=0$}
        {$Z\gets 1$ ;$T\gets 1$}
    }

\medskip
\tcp{Phase 4 (rounds $f$, $f+1$)}
\tcp{round $f$}
\texttt{Player $p$ always wakes up in round $f$}\\
\If{$Y=1$ or $Z=1$} {
    Send 1 to every player in $C_f$
}
\If{$p\in C_f$ and $p$ received value $1$} {
    $Y\gets 1$
}
\tcp{round $f+1$}
\texttt{Player $p$ always wakes up in round f+1}\\
\If{$p \in C_f$ and $Y=1$ \label{algline:binary_final_broadcast_condition}}
{Send $1$ to all\\}
\If{received value $1$}
{Decide 1 at the end of round $f+1$}
\Else {
Decide 0 at the end of round $f+1$\\
}
\end{multicols}
\vspace{.7em}
\caption{Binary Consensus. Code for player $p$, who has input value $X$.} 
\label{alg:binary}
\end{algorithm}

The algorithm is divided into four phases, and employs committees of players, of size roughly $\sqrt{n}$. During the first phase, players holding 1 as their initial value send it to the first committee. 
During the second phase, players who received the value 1 
become active for $\mathcal{O}(f/\sqrt{n})$ rounds, and propagate this value in each round to the 
currently awake committee.
Depending on the relation between $f$ and $\lfloor \sqrt{n} \rfloor ^2 - \lfloor \sqrt{n} \rfloor+1$, a third phase might happen. In this %
third phase, players behave the same as in phase 2, 
except that the committees are of size $f+1$ and the active players propagate the value 1 only if they received it during the last round. The fourth phase consists  of two last rounds. In round $f$, every player aware of the value 1 sends it to a committee of size $f+1$, and in the last round, each member of that committee broadcasts this value to all players if it received it on the previous round.
This strategy ensures that, assuming an adversary crashes at most $f$ players, if a player decides on 1 in the last round, then there must exist a non-faulty player holding 1 in some earlier round $< f+1$. This strategy also ensures a worst case energy complexity of $O(\ceil{f/\sqrt{n}})$.
A more detailed description follows next.

During initialization, we use Procedure \JoinCom\ (see Algorithm \ref{alg:joincommittees}) to fix $f$ different static committees $C_1,\ldots,C_f$.
Let $n' =\floor{\sqrt{n}}^2$ be the largest square number smaller or equal to $n$ and let $h = \min \{n'-\sqrt{n'}+1, f\}$. The size of committees depends on the relation between $f$ and $n'$. 
Specifically, we divide the committees into two parts: $C_1, \ldots, C_{h-1}$, where each committee is of size $\sqrt{n'}$, and $C_h, \ldots, C_f$, where each committee is of size $f+1$.
Note that if $f\leq n'-\sqrt{n'}+1$, then $h=f$ and $C_{h-1} = C_{f-1}$ (in other words, all the committees except for $C_f$ are of size $\sqrt{n'}$).
If $f>n'-\sqrt{n'}+1$, then there could be more than one committee of size $f+1$.
To simplify the notation, we also define $C_{f+1} = P$.
A crucial property of being a member of committee $C_r$ is that every $p \in C_r$ will wake up (at least) in round $r$, and thus we say that round $r$ is the \emph{awake round of committee $C_r$}.
Each player has two local variables $Y, Z$ that are initialized with 0, in which we store the player's current decision estimate (each variable is used in exactly one of two parts in the algorithm explained later).
Since our algorithm restricts itself to sending messages with bit value $1$, called \emph{$1$-valued message}, we only ever update $Y$ and $Z$ by setting $Y\gets 1$ or $Z\gets 1$, if they were $0$ before receiving such a message.
 
Conceptually, the algorithm consists of four phases:

Phase 1 consists only of round $1$, in which every player $p$ wakes up and, if $p$ has input value $X=1$, then it sends a message with a single bit $1$ to all members of committee $C_1$.
Player $p$ becomes \emph{active}, in the sense that $p$ remains awake for the next $T=\lt\lceil \frac{f+1}{\sqrt{n'}} \rt\rceil$ rounds, where $T$ is a counter that is decremented by $1$ in each subsequent round.
Any member $q \in C_1$ who has $Y=0$ and receives such a message becomes \emph{active} as well.

In Phase 2 (rounds $2, \ldots, h-1$), we propagate these $1$-valued messages across all other committees as follows. 
If a player $p' \in C_r$ receives the value $1$ for the first time upon waking up in round $r$ (i.e., $p'$ never received $1$ before nor had $1$ as their initial value), then $p'$ will become \emph{active}, and remain awake for the next $\ceil{\frac{f+1}{\sqrt{n'}}}$ rounds, during which $p'$ sends bit $1$ to the next corresponding committees whose members are awake in the respective round. Note that the length of Phase 2 depends on $h$ (or in other words the relation between $f$ and $n'$). If $f$ is relatively small ($f\leq n'-\sqrt{n'}+1$), then $h=f$ and upon completing Phase 2 in round $h-1=f-1$, the players proceed directly to Phase 4 (described below). Otherwise, the players proceed to Phase 3.

In Phase 3 (rounds $h, \ldots, f-1$), the size of the committees are $f+1$ and we propagate $1$-valued message in a similar way to Algorithm {\ref{alg:multi}}. During round $h$ (the first round of Phase 3), every player awakes and sends a $1$-valued message to every player in $C_h$, if it has $Y=1$.
Every player $p\in C_h$ that receives a $1$-valued message during round $h$ will awake during round $h+1$ and send a $1$-valued message to $C_{h+1}$. In subsequent rounds of Phase 3, every player $p\in C_r$ that receives a $1$-valued message for the first time in Phase 3 (it may have received a $1$-value message in previous phases) during round $r$ will awake in round $r+1$ and send a $1$-valued message to $C_{r+1}$.

In Phase 4 (rounds $f, f+1$), all players wake up. In round $f$,  every player that has $Y=1$ or $Z=1$ sends a $1$-valued message to every player in $C_f$. Every player in $C_f$ that receives a 1-valued message sets $Y=1$. In round $f+1$, a player $p$ broadcast 1 if $p\in C_f$ and $Y=1$. A player decides 1 if it receives a $1$-valued message during round $f+1$ and decides 0 otherwise. 

Figure \ref{fig:binary-flow} schematically illustrates the execution of Algorithm~\ref{alg:binary}.

\begin{figure}[t]
    \centering 
    \resizebox{0.9\linewidth}{!}{\begin{tikzpicture}
\tikzset{
player/.style={
    draw=gray,
    thin,
    dashed,
    opacity=0.6
},
message/.style={
    ->,
    ultra thick,
    >=Latex,
    shorten >=1pt, 
    shorten <=1pt, 
    rounded corners,
    color=cyan!70
},
cross/.style={
    path picture={ 
    \draw[black]
    (path picture bounding box.south east) -- (path picture bounding box.north west) 
    (path picture bounding box.south west) -- (path picture bounding box.north east);
    }
},
    time/.style={thick, dashed}
}

\def\t{15} %

\draw[thick, rounded corners] (0.5,-0.5) rectangle (2,10.5);
\node at (1.5, 11.3) {Phase 1};
\draw[thick, rounded corners] (2,-0.5) rectangle (7.5,10.5);
\node at (4.5, 11.3) {Phase 2};
\draw[thick, rounded corners] (7.5,-0.5) rectangle (13.5,10.5);
\node at (10.5, 11.3) {Phase 3};
\draw[thick, rounded corners] (13.5,-0.5) rectangle (15.5,10.5);
\node at (13.5, 11.3) {Phase 4};

\node[anchor=east] at (0,10) {$p_1$};
\draw[player] (0,10) -- (\t,10);
\node[anchor=east] at (0,9) {$\cdots$}; 
\draw[player] (0,9) -- (\t,9);
\node[anchor=east] at (0,8) {$p_{\sqrt{n'}+1}$};
\draw[player] (0,8) -- (\t,8);
\node[anchor=east] at (0,7) {$\cdots$};
\draw[player] (0,7) -- (\t,7);
\node[anchor=east] at (0,6) {$p_{2\sqrt{n'}+1}$};
\draw[player] (0,6) -- (\t,6);
\node[anchor=east] at (0,5) {$\cdots$};
\draw[player] (0,5) -- (\t,5);
\node[anchor=east] at (0,4) {$p_{3\sqrt{n'}+1}$};
\draw[player] (0,4) -- (\t,4);
\node[anchor=east] at (0,3) {$\cdots$};
\draw[player] (0,3) -- (\t,3);
\node[anchor=east] at (0,2) {$p_{4\sqrt{n'}+1}$};
\draw[player] (0,2) -- (\t,2);
\node[anchor=east] at (0,1) {$\cdots$};
\draw[player] (0,1) -- (\t,1);
\node[anchor=east] at (0,0) {$p_{n}$};
\draw[player] (0,0) -- (\t,0);

\node [draw,circle](pi1) at (1.5,9){$p_i$}; 
\node [draw,circle](pi2) at (2.5,9){$p_i$}; 
\node [draw,circle,cross](pi3) at (3.5,9){$p_i$}; 
\node [draw,circle](pj1) at (2.5,7){$p_j$}; 
\node [draw,circle](pj2) at (3.5,7){$p_j$}; 
\node [draw,circle,cross](pj3) at (4.5,7){$p_j$};
\node [draw,circle](pk1) at (3.5,5){$p_k$}; 
\node [draw,circle](pk2) at (4.5,5){$p_k$}; 
\node [draw,circle,cross](pk3) at (5.5,5){$p_k$}; 
\node [draw,circle](pl1) at (4.5,3){$p_l$}; 
\node [draw,circle](pl2) at (5.5,3){$p_l$}; 
\node [draw,circle,cross](pl3) at (6.5,3){$p_l$}; 
\node [draw,circle](pm1) at (5.5,1){$p_m$}; 
\node [draw,circle](pm2) at (8,1){$p_m$};
\node [draw,circle](pn1) at (8,8){$p_n$};
\node [draw,circle](pn2) at (9,8){$p_n$};
\node [draw,circle](po1) at (9,6){$p_o$};
\node [draw,circle](po2) at (10,6){$p_o$};
\node [draw,circle](pp1) at (10,3){$p_p$};

\node [draw,circle](pa) at (14,4.5){$p_a$};

\node [draw,circle](pa2) at (15,4.5){$p_a$};

\draw[thick, dashed, rounded corners, fill=blue!20, fill opacity=0.3] (0.6, 10.1) rectangle (1.9, 8.1);
\node[text=blue] at (1, 8.5) {$C_1$};

\draw[thick, dashed, rounded corners, fill=blue!20, fill opacity=0.3] (2.1, 8.1) rectangle (2.9, 6.1);
\node[text=blue] at (2.5, 6.4) {$C_2$};

\draw[thick, dashed, rounded corners, fill=blue!20, fill opacity=0.3] (3.1, 6.1) rectangle (3.9, 4.1);
\node[text=blue] at (3.5, 4.4) {$C_3$};

\draw[thick, dashed, rounded corners, fill=blue!20, fill opacity=0.3] (4.1, 4.1) rectangle (4.9, 2.1);
\node[text=blue] at (4.5, 2.4) {$C_4$};

\draw[thick, dashed, rounded corners, fill=blue!20, fill opacity=0.3] (5.1, 2.1) rectangle (5.9, 0.5);
\node[text=blue] at (5.5, 0.3) {$C_{...}$};

\draw[thick, dashed, rounded corners, fill=blue!20, fill opacity=0.3] (7.6, 10.1) rectangle (8.4, 7.1);
\node[text=blue] at (8, 9.5) {$C_{h}$};

\draw[thick, dashed, rounded corners, fill=blue!20, fill opacity=0.3] (8.6, 7.1) rectangle (9.4, 4.1);
\node[text=blue] at (9, 5.3) {$C_{h+1}$};

\draw[thick, dashed, rounded corners, fill=blue!20, fill opacity=0.3] (9.6, 4.1) rectangle (10.4, 1.1);
\node[text=blue] at (10, 2.3) {$C_{h+2}$};

\draw[thick, dashed, rounded corners, fill=blue!20, fill opacity=0.3] (13.6, 6.1) rectangle (14.4, 3.1);
\node[text=blue] at (13, 5.5) {$C_{f}$};

\draw[thick, dashed, rounded corners, fill=blue!20, fill opacity=0.3] (14.6, 10.1) rectangle (15.4, -0.1);
\node[text=blue] at (16, 9.5) {$C_{f+1}$};

\foreach \step in {10,7,5,1}{\draw[message,shorten >=0px,-{}] (0.5, \step) -- (1.5, \step);}
\foreach \step in {10,1}{\draw[message,shorten >=10px] (1.5, \step) -- (1.5, 9);}
\draw[message,dashed] (pi2) -- (pj1);
\draw[message,dashed] (pj2) -- (pk1);
\draw[message,dashed] (pk2) -- (pl1);
\draw[message,dashed] (pl2) -- (pm1);
\draw[message,dashed] (pm2) -- (pn1);
\draw[message,dashed] (pn2) -- (po1);
\draw[message,dashed] (po2) -- (pp1);

\draw[message,shorten >=0px] (14,10) -- (14, 6);
\draw[message,shorten >=0px] (14,0) -- (14, 3.1);
\draw[message,shorten >=0px,-{}] (15, 0) -- (pa2);
\draw[message,shorten >=0px,-{}] (pa2) -- (15, 10);
\foreach \step in {0,...,10}{\draw[message] (15, \step) -- (15.5, \step);}

\end{tikzpicture}}
    \caption{An Example of the Execution of Binary Consensus (Algorithm~\ref{alg:binary}):\\
    Phase 1: Every player with input value $X=1$ sends a message to $p_i$, then crashes.\\
    Phase 2: $p_i$ is only able to send a message to $p_j \in C_2$ as it crashes. The same happens in each of the rounds of Phase~2, i.e., the player sending $1$ can only send its message to one member of the next committee. \\ 
    Phase 3: Assume that $f > n' - \sqrt{n'}$. All committees are now of size $f+1$. Player $p_m$, who received the message in the last round of Phase~2, tries to send a message to everyone in $C_h$ but crashes and the message is only received by $p_n$; in the remaining rounds of Phase~3, this process continues analogously.\\
    Phase 4: Every alive player with $Y=1$ try to send messages in $C_f$ in round $f$, they may crash before fully propagating 1 to $C_f$. $p_a\in C_f$ receives a message. In round $f+1$, $p_a$ try to broadcast $1$, and the adversary here can not crash $p_a$ without exceeding the threshold $f$. Consensus is reached. 
   \label{fig:binary-flow}
   }
\end{figure}

The main result of this section, proved hereafter, is the following.
\begin{theorem} 
\label{thm:binary}
Suppose that each player starts with a one-bit input value and that at most $f<n$ players crash. There is a deterministic algorithm that solves consensus in $f+1$ rounds, sends
$\mathcal{O}(nf)$ messages of constant size. 
Moreover, the algorithm achieves 
$\mathcal{O}(\lceil{f}/{\sqrt{n}}\rceil)$ energy complexity.
\end{theorem}

\subsection{Notation and  Terminology}
Before proving Theorem~{\ref{thm:binary}}, we need to define some crucial terms for the  analysis of the algorithm.
Recall that 
$h = \min \{n'-\sqrt{n'}+1, f\}$.
If $f \le n'-\sqrt{n'}$, then Phase 3 is not executed, and
$h=f$
is the first round of the last phase of the execution (denoted Phase 4 in the algorithm, for rounds $f,f+1$). Otherwise, $h=n'-\sqrt{n'}+1$ is the round that starts Phase 3.\\
Consider some round $r$. We say that a \emph{player $p$ is triggered in round $r$}, if $p$ receives a 1-valued message in round $r$ that causes it to set its timer (and start broadcasting 1 as long as its timer is not 0), or $r \in \{1, h, f+1\}$ and $p$ wakes up in $r$ to send a 1-valued message. %
The next observation is immediate from the description of the algorithm.
\begin{observation} \label{obs:trigger_round}
Let $m$ be a specific 1-valued message sent by $p$. There exists a round $r$ such that one of the following is true:
\begin{itemize}
\item $m$ is sent as a consequence of $p$ being triggered in round $r$ by receiving a 1-valued message from some player $p'$, in some round following $r$,
or
\item $r\in \{1,h, f+1\}$ and $p$ wakes up during round $r$ by the algorithm's description and sends $m$ (since it has $Z=1$, $Y=1$ or $X=1$, depending on the round). 
\end{itemize}
\end{observation}
We refer to the round $r$ guaranteed by Observation~\ref{obs:trigger_round} as the \emph{trigger round of $m$}.
Note that if $r\in \{h, f+1\}$ and $p$ wakes up by the algorithm's description during $r$ (for having $Y=1$), then $p$ must have been triggered in some prior round $r'<r$, in which case the trigger round of $m$ is $r'$.
We say that $p$ was \emph{triggered by $p'$} to send $m$ in the former case, as well as in the latter case if $r\in \{h,f+1\}$.
We say that $p$ was \emph{triggered by the algorithm} before round $1$ to send $m$ in the latter case if $r=1$. For convenience, we define $C_{f+1} = V$ (i.e., the set of all players).

\begin{observation}
\label{obs:triggered_once}
The following hold:
\begin{enumerate}
\item[(1)] During each round $r\in [h+1, f+1]$, a player sends a 1-valued message only if it was triggered in the previous round, and that can happen at most once. 
\item[(2)] A player can be triggered to send messages during rounds $[1,h]$ only once.
\end{enumerate}
\end{observation}
\begin{proof} 
Part (1) follows since a player can only be triggered if $Z=0$, and when it is triggered it sets $Z=1$. Part~(2)  holds since a player can only be triggered if $Y=0$, upon which $Y$ is set to 1.
\end{proof}
\subsection{Proof of Theorem \ref{thm:binary}}
To prove the correctness of the algorithm, we need to show that it satisfies the validity, termination and agreement requirements. 
To start the proof, we make the following key observation on the structure of committees. Note that this observation is the reason we use the first $n'$ players for the committees of Phase 2 rather then all $n$ players.
\begin{observation} \label{obs:comm_diff}
$C_i = C_{i+\sqrt{n'}}$ for every 
$i\in [1,h-1-\sqrt{n'}]$.
\end{observation}
\begin{proof}
This follows from the round robin construction of the committees in Algorithm {\ref{alg:joincommittees}} and the fact that $\sqrt{n'}$ is an integer; i.e, we can form $\sqrt{n'}$ committees of size $\sqrt{n'}$. Then, committee number $\sqrt{n'} + 1$  will be the same as committee number 1 and so on.
\end{proof}

Next, we show that if enough players ($f+1$) receive a $1$-valued message prior to round $f+1$, then agreement on $1$ is guaranteed.

\begin{claim}
    \label{clm:non-faulty_received_1}
    If a non-faulty player receives a 1-valued message before round $f+1$, then all non-faulty players decide $1$ at the end of the execution.
\end{claim}
\begin{proof}
    Let $p$ be a non-faulty player that receives a 1-valued message for the first time during round $r\leq f$. If $r=f$, then $p\in C_r$ and during round $f$, $p$ sets $Y\gets 1$. Hence, in round $f+1$, $p$ broadcasts 1 to all players, resulting in every non-faulty processor deciding 1. Otherwise ($r<f$), $p$ awakes in round $f$ and sends 1 to every player in $C_f$. Since $|C_f|=f+1$, at least one recipient is non-faulty. Following the same reasoning of the previous case, the claim follows. 
\end{proof}
A direct implication of Claim {\ref{clm:non-faulty_received_1}} is the following corollary.
\begin{corollary}\label{cor:f+1_received}
If $f+1$ players receive a 1-valued message before round $f+1$, then all non-faulty players decide 1 at the end of the execution.
\end{corollary}
\begin{proof}
At least one of those $f+1$ players is non-faulty, and thus %
every non-faulty player decides 1 by Claim \ref{clm:non-faulty_received_1}.
\end{proof}

Since we defined a player to be non-faulty if it does not crash before deciding, it suffices to show validity, termination, and agreement for non-faulty players.

\begin{lemma} \label{lem:binary_validity_term}
Algorithm \ref{alg:multi} satisfies validity and terminates in $f+1$ rounds.
\end{lemma}
\begin{proof}
Every non-faulty player decides on $Y$ in round $f+1$, hence termination is guaranteed in $f+1$ rounds.

Next, we argue validity:
In the case where every player starts with $0$, no message is sent throughout the entire execution, causing all players who decide to decide on their input value.
Now consider the case where everyone starts with $1$. Since $n\geq f+1$, then by Corollary {\ref{cor:f+1_received}}, every non-faulty player decides 1, and validity follows. 
\end{proof}

We now focus on the agreement property:
The cases where every player starts with the same input value are already covered by Lemma \ref{lem:binary_validity_term}, and henceforth we can assume that there are players starting with $0$ as well as with $1$.

We now prove the main correctness claim, namely, that the agreement requirement is satisfied by the algorithm.
Assume towards contradiction that there exist an execution $EX$ and two non-faulty players $p^0,p^1$ that decide $0,1$ respectively in $EX$. By the algorithm's description, $p^1$ must have received a 1-valued message during round $f+1$ (not before, otherwise by Claim \ref{clm:non-faulty_received_1}, $p^0$ would decide 1). Denote by $\p_\ell$ 
the player that sent $p^1$ the 1 valued message during round $f+1$ (if there are more than one, then choose the one that was triggered to do so first, and if there are still more than one, choose one arbitrarily). Note that $\p_\ell$ must have been triggered (to send that message) by some player $\p_{\ell-1}$ during some prior round $r_{\ell-1}$ (the same rules for picking $\p_{\ell-1}$ apply). In general, $\p_{j+1}$ is either triggered by $\p_j$ during round $r_j$ or is triggered by the algorithm before round $1$.
We obtain a sequence of players $\p_1,\dots,\p_{\ell+1}$, where $\p_1$ is 
some player that has input 1 (and is triggered by the algorithm before round $1$), $\p_{\ell+1}=p^1$, and $r_\ell=f+1$. For ease of notation we will refer to $\p_j$ sending 1 to $\p_{j+1}$ as \emph{$\p_j$ triggering $\p_{j+1}$}.

\begin{observation}
\label{obs:phase_3_no_gaps}
    For every round $r\in [h+1, f+1]$, $r=r_j$ for some $j\in \ell$.
\end{observation}
\begin{proof}
    We prove the observation inductively on the round number $r$.

    \textbf{Base ($r=f+1$)}: $r=f+1=r_\ell$ by definition.

    \textbf{Step $(r<f+1)$}:
    By the induction hypothesis, we know that for $r+1 = r_j$ for some $j\le \ell$. Recall that $\p_{j}$ is triggered by $\p_{j-1}$ during round $r_{j-1}<r_j=r+1$ ($r_{j-1}\leq r$) and that $\p_j$ triggers $\p_{j+1}$ during round $r_j=r+1$. If $r>r_{j-1}$, then $\p_j$ does not crash during round $r$ and sends 1 to every player in $C_r$. By Corollary \ref{cor:f+1_received}, every non-faulty player decides 1, since $|C_r|\geq f+1$.
\end{proof}
\begin{observation}
\label{obs:phase2_no_3}
In execution $EX$, a player that is triggered in round $r<h$ cannot be triggered during round $r'\geq h$.
\end{observation}
\begin{proof}
Assume towards contradiction that $p$ is triggered during both rounds $r<h$ and $r'\geq h$.
Hence, $p$ does not crash during round $h$, and since it was triggered during round $r<h$, it has $Y=1$ during round $h$, which means it sends 1 to $C_h$.
If $h=f+1$, then $p^0$ receives and decides 1; otherwise, by Corollary \ref{cor:f+1_received} ($|C_h|=f+1$), it follows again that $p^0$ decides 1, which is a contradiction.
\end{proof}

\begin{claim} 
\label{clm:unique_players}
Every player in the sequence $\p_1, \dots, \p_\ell$ is unique.
\end{claim}
\begin{proof}
Assume towards a contradiction that $\p_j = \p_{j'}$ for some indices $j< j'$.
Consider the following cases.
If $h\leq r_j$ (Phase 3), then $\p_j=\p_{j'}$ sends a message during round $r_{j'}$, which means that it did not crash during round $r_j$ and sent 1 to every player in $C_{r_j+1}$. In Phase 3, the committee size is $|C_{r_j}| = f+1$.
Therefore, by Corollary \ref{cor:f+1_received}, all non-faulty players will decide 1, in contradiction.

If $r_j< h$ (Phase 1 or 2), then $\p_j$ is triggered during round $r_{j-1}$ and $\p_{j'}$ is triggered during round $r_{j'-1}$. Note that  $r_{j-1}<r_j\leq r_{j'-1}$. It cannot be that $r_{j'-1} \geq h$ because otherwise $\p_j=\p_{j'}$ is triggered in both $r_{j-1}< h$ and $r_{j'-1} \geq h$ violating Observation \ref{obs:phase2_no_3}. Also, it cannot be that  $r_{j'-1}< h$, because a player can be triggered only once during rounds $[1,h]$.
Hence, $r_{j'-1}$ cannot exist.
\end{proof}

\begin{claim}
\label{clm:crash_before_timer}
Every player that sets its timer $T>0$ during execution $EX$ must have crashed before it sets $T\gets 0$.
\end{claim}
\begin{proof}
Let $p$ be a player that sets its timer during round $r\leq f$. Consider two cases.

If $r\geq h$ and $h<f+1$ (Phase 3), then $p$ sets its timer to 1. In the following round, $r+1$, $p$ sends 1 to $C_{r+1}$. If $r+1\leq f$ then $|C_{r+1}|=f+1$ and  by Corollary \ref{cor:f+1_received}, every non-faulty player decides 1. If $r+1= f+1$, then $p\in C_f$, $|C_{r+1}|=n$ and every non-faulty player receives 1 and thus decides 1. In both cases all non-faulty players decide 1, contradiction.

If $r< h$ (Phase 1 or 2), then $p$ sets its timer to $\ceil{(f+1)/\sqrt{n'}}$. If $p$ does not crash before its timer reaches 0, then it sends 1 to the committees awake during the $\ceil{\frac{f+1}{\sqrt{n'}}}$ rounds following $r$.
Since the size of every such committee is at least $\sqrt{n'}$ and by Observation \ref{obs:comm_diff}, $p$ sends 1 to at least $\min \{f+1, n'\}$ different players. If $f\leq n'-\sqrt{n}$ or $p$ sends 1 during round $h$, then $p$ sends 1 to at least $f+1$ players, and by Corollary \ref{cor:f+1_received}, every non-faulty player decides 1, contradiction. 

The remaining, more difficult case is when $f>n'-\sqrt{n'}$ and $p$ does not send 1 during round $h$.
In this case, $p$ sends 1 to all of the first $n'<f+1$ players and $h= n'-\sqrt{n'}+1$. All of those players must crash before round $h+1$ (otherwise, every player in $C_h$ receives 1 and Corollary \ref{cor:f+1_received} applies).
Recall that for every round $r\in [h+1,f+1]$, $r=r_j$ for some $j\in[1,\ell]$ (Observation \ref{obs:phase_3_no_gaps}) and there is a unique player $\p_j$ that sends 1 during round $r$ (by Claim \ref{clm:unique_players}). The player $\p_j$ cannot have been triggered during any round $r'< h$ by Observation \ref{obs:phase2_no_3}, so $\p_j\notin \{p_1, \ldots,p_{n'}\}$ (as they were triggered during a round in $[1,h-1]$). If $\p_j$ is non-faulty, then its timer reaches 0 and every non-faulty player decides 1 (as explained above). Hence, the set of players $\{p_1, p_2, \dots, p_{n'}\} \cup \{\p_j \mid r_j \in [h+1,f+1]\}$ contains exactly $n' + f-h+1 = f +\sqrt{n'}\geq f+1$ unique players that must crash, contradiction.
\end{proof}

\begin{observation} \label{obs:inner_rounds}
    In every round $r$ such that $r_j<r<r_{j+1}$ for some $j\in [1,\ell]$, every player in the committee $C_r$ receives a 1-valued message.
\end{observation}
\begin{proof}
    By definition, during round $r_{j+1}$ the player $\p_{j+1}$ sends 1 to $\p_{j+2}$ and during round $r_j$, $\p_{j+1}$ is triggered by $\p_j$. Hence, in all rounds $r_j<r<r_{j+1}$, $\p_{j+1}$ does not crash (otherwise it will not be able to send during round $r_{j+1}$) and must have a non-zero timer, or in other words send a message, by Claim \ref{clm:crash_before_timer}. Hence, every players in $C_r$ receives 1 from $\p_{j+1}$.
\end{proof}

A \emph{failure mapping} for execution $EX$ is a mapping $\varphi : [0,f] \rightarrow V$ that maps every round $j\in [0,f]$ to a player $\varphi(j)$ such that 
\begin{enumerate}
    \item[(P1)] For $j\in[1,f]$, $\varphi(j) \in C_j$ and $\varphi(j)$ has crashed in  $EX$.
    \item[(P2)] $\varphi(0) = \p_1$.
    \item[(P3)] $\varphi$ is an injection, i.e., the players $\varphi(j)$ are distinct.
\end{enumerate}
Note that the existence of a failure mapping for execution $EX$ leads to contradiction, since the combination of these properties means that every round $j\in[0,f]$ is associated with a unique player that crashes in $EX$ (otherwise, $p^0$ decides 1), which is $f+1$ unique players, in contradiction to the number of faults being $f$.

We now describe a 
procedure $\Constr$ that, given an execution $EX$, attempts constructing a corresponding failure mapping $\varphi$ for $EX$. Procedure $\Constr(EX)$ operates as follows:

\begin{enumerate}
\item
Let $R=\{r_j\mid j\in [1, \ell-1]\}$ and $P =\{\p_j\mid j\in[1,\ell]\}$.
\item $\varphi(0) = \p_1$ (satisfies property (P2)).
\item For $r_j\in R$, $\varphi(r_j) = \p_{j+1}$.
\item For $r\in [1,\dots, f]\setminus R$, we define $\varphi(r) = p\in C_r$, where $p$ is chosen such that there is no $r'< r$ with $\varphi(r') = p$.
\end{enumerate}

Note that  a-priori, Procedure $\Constr(EX)$ might fail to produce a failure mapping for $EX$. This might happen if at some round $r \notin R$, there is no player in $C_r$ that is not already used in a previous round (since the players $\p_j$ for $j\in[1,\ell]$ are $m$ unique players by claim \ref{clm:unique_players}). Note that $r<h$.

\begin{claim}
If Procedure $\Constr(EX)$ is completed successfully, then Properties (P1), (P2), (P3) are satisfied.
\end{claim}
\begin{proof}
Property (P1) is satisfied by the construction since for every round $r_j\in R$, by definition, $\p_{j+1}\in C_{r_j}$  receives a 1-valued message, and, for every round $r\in [1,f]\setminus R$, every player $p\in C_r$ receives a 1-valued message, by Observation \ref{obs:inner_rounds}. 
    
Property (P2) is satisfied since $\varphi(0) = \p_1$ by construction.
    
Property (P3) is satisfied because of the following argument.
Let $r<r'$ be two rounds. If $r' \in \{1,\dots, f\}\setminus R$, then by construction $\varphi(r) \neq \varphi(r')$. Otherwise ($r'\in R$), if $r \in R$ then $\varphi(r) \neq \varphi(r')$ by Claim \ref{clm:unique_players}. If $r \in \{1,\dots, f\}\setminus R$ and $r' = r_j$ for some $j\in[1,\ell-1]$ then $\varphi(r')= p_{j+1}$. $p_{j+1} \notin C_r$ since (1) $r<h$ (2) $\p_{j+1}$ receives 1 and sets its timer during round $r_j$ (3) if $r_j \geq h$, then during round $h$, $\p_{j+1}$ sends 1 to $C_h$ in contradiction. Otherwise $r_j<h$, but a player can set its timer only once during round $[1,h-1]$, a contradiction. Hence, $p_{j+1} \neq \varphi(r)\in C_r$.
\end{proof}

To complete the proof, we argue that this does not happen, i.e., the procedure always succeeds (hence leading to contradiction as discussed earlier).

\begin{claim}
Procedure $\Constr(EX)$
is always successful.
\end{claim}
\begin{proof}
Assume towards contradiction that the construction is unsuccessful. Hence, there exists some round $r\notin R$, such that $r\leq \min \{f, n'-\sqrt{n'}\}=h-1$, where there is no player in $C_r$ that is not already used in a previous round.
Recall that $|C_r| = \sqrt{n'}$. For every $p\in C_r$ there exists $r' < r$ such that $\varphi(r')=p$. Let $p'_1, p'_2, \dots p'_{\sqrt{n'}} \in C_r$ be the players in $C_r$ in order of round labeling, and let $r(p'_j)<r$ be the round such that $\varphi(r(p'_j)) = p'_j$ for $j\in [1,\sqrt{n'}]$. It might be the case that $p'_1=\p_1$ in which case $r(p'_1) = 0$.
Hence,  $r(p'_2)\geq 1$ and by Observation \ref{obs:comm_diff}, $r(p'_{j+1}) \geq r(p'_{j}) + \sqrt{n'}$. Therefore, by induction we get that $r \geq r(p'_{\sqrt{n'}}) + \sqrt{n'} \geq 1 +(\sqrt{n'}-1)\sqrt{n'} = 1 + n'-\sqrt{n'} > \min \{f, n'-\sqrt{n'}\}=h-1$ in contradiction.
\end{proof} 

The last claim implies that the existence of execution $EX$ leads to the existence of a failure mapping for it, in contradiction with the assumption of at most $f$ failures. This yields the following.

\begin{lemma}
Algorithm~\ref{alg:binary} (Binary Consensus) achieves agreement.
\end{lemma}

\begin{lemma}\label{lem:binary_complexity}
Algorithm \ref{alg:binary} sends $\mathcal{O}(nf)$ messages of constant size and has an energy complexity of $\mathcal{O}(\lceil{f}/{\sqrt{n}}\rceil)$.
\end{lemma}
\begin{proof} 
We start by proving the message complexity. Note that if $f\leq \sqrt{n}$, then we can simply use Algorithm~{\ref{alg:multi}}, which, in that case, already achieves awake complexity $O(1)$, round complexity $f+1$ and message complexity $\mathcal{O}(fn)$; note that the players know $f$ and $\sqrt{n}$ and hence can locally determine which algorithm to execute. Hence, in the remainder of the analysis, we can assume that $f> \sqrt{n}$.
\begin{description} 
    \item[(Phase 1):] Every player with input value $1$ may send a message to committee $C_1$. Therefore $\mathcal{O}(n\cdot \sqrt{n'}) =\mathcal{O}(nf)$ message might be sent since $|C_1|=\sqrt{n'}$.
    \item[(Phase 2):] 
    Every player $p$ can be triggered only once during rounds $[1,h]$ (Observation {\ref{obs:triggered_once}}), resulting in $p$ sending messages to the current round's committee for $\ceil{\frac{f+1}{\sqrt{n'}}}$ rounds, and the size of a committee in phase 2 is $\sqrt{n'}$. Hence,
    the message complexity of this phase is $\mathcal{O}(nf+n\sqrt{n}) = \mathcal{O}(nf)$.
    
    \item[(Phase 3):] In Phase 3, every player may send messages in round $r$ to all players in $C_{r+1}$ at most once. Since every committee in Phase 3 is of size $f+1$, this accounts for at most $\mathcal{O}(nf)$ messages.
    \item[(Phase 4):] In round $f$, every player may send a message to every player in $C_f$. In the final round, up to $f+1$ players may broadcast, which again yields $\mathcal{O}(nf)$ messages.
\end{description}

\noindent Next, we analyze the energy complexity:
\begin{description}
    \item[(Phase 1 and 4):] The first and last phase consist of a constant number of rounds, which amounts to  $\mathcal{O}(1)$ awake rounds per node. 
    \item[(Phase 2):] Recall that a player is awake in round $r$ if its timer is not 0 or it belongs to $C_r$. During the rounds $1\leq r< h$, the timer of a player $p$ can be set to $\ceil{(f+1)/\sqrt{n'}} = \mathcal{O}(f/\sqrt{n})$ only once. Therefore, every player is awake in at most 
    $$O\lt(\frac{\min \{f, n'-\sqrt{n'}\} \cdot\sqrt{n'}}{n'}+\lt\lceil\frac{f+1}{\sqrt{n'}}\rt\rceil \rt)=\mathcal{O}(f/\sqrt{n})$$
    rounds.
    \item[(Phase 3):] During rounds $r\geq h$, the timer of a player $p$ can be set to 1  once, and it belongs to. 
    \[O\lt(\frac{\max \{0, f-(n'-\sqrt{n'})\} \cdot (f+1)}{n}\rt)=\mathcal{O}(f\sqrt{n}/n)=\mathcal{O}(f/\sqrt{n})\]
    committees, where the first equality follows since
    \[
    f-(n'-\sqrt{n'}) \leq n - \sqrt{n'} (\sqrt{n'}-1)\leq n-(\sqrt{n}-1)(\sqrt{n}-2)=3\sqrt{n}-2
    \]
\end{description}
Hence, the awake complexity is $\mathcal{O}(\ceil{f/\sqrt{n}})$.
\end{proof}

\section{Conclusion}

In this work, we have initiated the study of fault-tolerant and energy-efficient algorithms in synchronous networks, by showing that the $f+1$ barrier on the time complexity of deterministic consensus algorithms does not apply to energy complexity.
An intriguing question raised by our work is how close the energy complexity bound of $\Theta( f/\sqrt{n}  )$ is to the lower bound and, more specifically, whether the energy complexity must depend on the number of faults $f$.
Note that it is possible to solve certain graph problems such as vertex coloring on a ring in just $O(1)$ energy~\cite{DBLP:conf/podc/BalliuFOR25}, even though they are known to have a super-constant time complexity~\cite{linial1992locality}, which means that we cannot rule out an $O(1)$-energy algorithm for consensus. 

\bibliographystyle{plainurl}
\bibliography{refs}

\end{document}